\documentclass[twocolumn,preprintnumbers,amsmath,amssymb]{revtex4}
\usepackage{graphicx}
\usepackage{dcolumn}
\usepackage{bm}
\usepackage[all]{xy}
\usepackage{indentfirst}
\usepackage{multirow}
\usepackage{mathrsfs}
\usepackage{bbm}
\usepackage{euscript}
\usepackage{extarrows}
\usepackage{epstopdf}
\usepackage{amsmath, amsthm, amssymb}
\usepackage[urlcolor=blue,citecolor=blue,linkcolor=red,colorlinks=true]{hyperref}
\usepackage{verbatim}
\newcommand{\ketbra}[2]{|#1\rangle \langle #2|}
\newcommand{\ket}[1]{|#1\rangle}
\newcommand{\bra}[1]{\langle #1|}
\newcommand{\Tr}{\textrm{Tr}}
\newcommand{\ham}{\hat{H}}
\newcommand{\ergo}{\mathcal{E}}
\newcommand{\antiergo}{\mathcal{A}}
\newcommand{\workrange}{\mathcal{C}}

\newtheorem{definition}{Definition}
\newtheorem{prop}{Proposition}

\begin{document}

\title{The Battery Capacity of Energy-storing Quantum Systems}

\author{Xue Yang$^{1,2\, *}$, Yan-Han Yang$^1$
\footnote{These authors contributed equally to this work.}, Mir Alimuddin$^{3}$, Raffaele Salvia$^{4}$, Shao-Ming Fei$^{5,6}$
\footnote{feishm@cnu.edu.cn}, Li-Ming Zhao$^1$, Stefan Nimmrichter$^{7}$, Ming-Xing Luo$^{1,8,9}$
\footnote{mxluo@swjtu.edu.cn}}

\affiliation{1. The School of Information Science and Technology, Southwest Jiaotong University, Chengdu 610031, China;
\\
2. School of Computer and Network Security, Chengdu University of Technology, Chengdu 610059, China;
\\
3. Department of Physics of Complex Systems, S. N. Bose National Center for Basic Sciences, Kolkata 700106, India;
\\
4. Scuola Normale Superiore, I-56127 Pisa, Italy;
\\
5. School of Mathematical Sciences, Capital Normal University, Beijing 100048, China;
\\
6. Max-Planck-Institute for Mathematics in the Sciences, 04103 Leipzig, Germany;
\\
7. Naturwissenschaftlich-Technische Fakult\"{a}t, Universit\"{a}t Siegen, Siegen 57068, Germany;
\\
8. CAS Center for Excellence in Quantum Information and Quantum Physics, Hefei, 230026, China;
\\
9. Shenzhen Institute for Quantum Science and Engineering, Southern University of Science and Technology, Shenzhen 518055, China
}

\begin{abstract}
The quantum battery capacity is introduced in this letter as a figure of merit that expresses the potential of a quantum system to store and supply energy. It is defined as the difference between the highest and the lowest energy that can be reached by means of the unitary evolution of the system. This function is closely connected to the ergotropy, but it does not depend on the temporary level of energy of the system. The capacity of a quantum battery can be directly linked with the entropy of the battery state, as well as with measures of coherence and entanglement.
\end{abstract}

\maketitle

Quantum thermodynamics is a blossoming field that aims to bridge the gap between quantum physics and thermodynamics. The growing interest in quantum technologies has created fertile ground for the theoretical and experimental study of quantum batteries, i.e., of quantum devices that can store and release energy in a controllable manner \cite{Perarnau2015,Vin2016,Ciampini2017,Francica2017,Andolina2019,Monsel2020,Opatrny2021}. Thanks to their capability of exploiting coherence, quantum batteries could facilitate faster, higher-power charging than their classical counterparts.

A central quantity in the study of quantum batteries is the ergotropy \cite{Allahverdyan2004}, which represents the amount of energy that can be extracted from a given quantum battery state by means of cyclic modulations of the battery's Hamiltonian (or, equivalently, by unitary evolution). As the battery releases or stores energy \cite{Binder2015,Campaioli2017,Campaioli2018,Rossini2020,Salvia2021,Seah2021,Shaghaghi2022a,Salvia2022,Francica2022,Rodriguez2023}, its ergotropy may change from zero (in which case the battery is said to be in its zero-charge \emph{passive state}) \cite{PS1977,PS1978}, to a maximum value $\workrange(\hat\varrho; \ham)$ that can be calculated from the eigenvalues of the battery's density matrix $\hat\varrho$ and from the energy levels of the Hamiltonian $\ham$.

In this paper, we discuss the \emph{quantum battery capacity} $\workrange(\hat\varrho; \ham)$ as a figure of merit linking its work storage capacity to quantum features such as  quantum entropies \cite{Neumann,Tsallis,HHH}, or quantum coherences \cite{Aberg,South,Streltsov,Bu,Baumgratz}.
Although most of the properties of the battery capacity can be derived from the properties of the ergotropy, we argue that the battery capacity is in some sense a more fundamental quantity as it does not change when the battery is unitarily charged or discharged. Furthermore, at variance with the ergotropy, as a spectral functional of state $\hat\varrho$ and Hamiltonian $\ham$, the battery capacity can be a simpler quantity from a theoretical point of view. The fact that the battery capacity depends on the state only through its eigenvalues makes it easy to operationally connect it with entropy and coherence measure for general battery system with equally spaced energy levels.

More recently, composite quantum systems have been considered for work storage \cite{Mukherjee2016,Nielsen2000,Alicki2013,Hovhannisyan2013,Alimuddin2019,Alimuddin2020,Gyhm2022,Puliyil2022,Andolina2019,Delmonte2021}, tapping into the resource of quantum entanglement. The amount of work that can be extracted from a composite quantum system is usually bigger if we are allowed to perform global operations on the system, than if we can only act locally on its subsystems. This advantage is reflected in a bigger value of the ergotropy (called the \emph{ergotropic gap} \cite{Mukherjee2016,Puliyil2022}), and in different statistics of work extraction with respect to random unitary transformations \cite{Imai2023,Shaghaghi2022}. This advantage of global operations is also reflected in a gap in battery capacity; here we show, inspired by the known results for the ergotropic gap, that also the battery capacity gap can serve as a witness of bipartite entanglement and genuine multipartite entanglement.


\begin{figure}
\includegraphics[width=0.8\columnwidth]{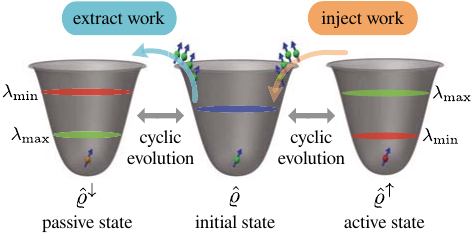}
\caption{\small Pictorial representation of the charging and discharging of a two-level quantum battery under cyclic evolution. Given an initial state $\hat\varrho$ with eigenvalues $\lambda_{\min}$ and $\lambda_{\max} \geq \lambda_{\min}$, the energy of the battery can vary between energy of the passive state $\lambda_{\min}E$ and the energy of the active state $\lambda_{\max}E$.}
\label{inject}
\end{figure}

\textit{Extracting and injecting work in a quantum battery.}---Consider an isolated $d$-dimensional quantum battery system equipped with a bare Hamiltonian $\ham$ that determines its energy spectrum and an initially prepared state $\hat\varrho$ that determines how much useful energy charge the battery can carry. Our aim is to assess the amount of charge that can be added or removed from the battery in control protocols that do not involve heat exchange with a thermal environment.

When the battery is subjected to a cyclic driving of the system Hamiltonian, its state undergoes an unitary evolution $\hat\varrho \to \hat{U}\hat\varrho\hat{U}^\dagger$ and its mean energy changes by
\begin{eqnarray}
W_{\hat{U}}(\hat\varrho; \ham ) \equiv \Tr[\hat\varrho\ham] - \Tr[\hat{U}\hat\varrho\hat{U}^\dagger \ham],
\label{def_WU}
\end{eqnarray}
which we identify, following the paradigm introduced in \cite{PS1977, PS1978}, with the amount of work \textit{extracted} from the battery.
The work extraction functional~(\ref{def_WU}) is bounded by the inequalities
\begin{eqnarray}
\ergo(\hat\varrho; \ham ) \geq W_{\hat{U}}(\hat\varrho; \ham ) \geq \antiergo(\hat\varrho; \ham)\;,
\label{bounds_Wu}
\end{eqnarray}
where the quantities $\ergo(\hat\varrho; \ham)$ and $\antiergo(\hat\varrho; \ham)$ are called the \emph{ergotropy} and the \emph{antiergotropy} \cite{Salvia2021} of the quantum state $\hat\varrho$ with respect to the Hamiltonian $\ham$:
\begin{eqnarray}
&&\ergo(\hat\varrho; \ham ) \equiv \max_{\hat{U} \in \mathbf{U}(d)} W_{\hat{U}}\left(\hat\varrho; \ham \right)\;;
\label{def_ergotropy} \\
&&\antiergo(\hat\varrho; \ham) \equiv \min_{\hat{U} \in \mathbf{U}(d)} W_{\hat{U}}(\hat\varrho; \ham )\; .
\label{def_antiergotropy}
\end{eqnarray}
Here, $\mathbf{U}(d)$ represents the unitary group of $d\times d$ matrices.

Let $\hat{U}^{(\downarrow)}$ and $\hat{U}^{(\uparrow)}$ denote, respectively, the unitary transformation that realize the maximum~(\ref{def_ergotropy}) and the minimum~(\ref{def_antiergotropy}) of the work extraction functional. The state $\hat\varrho^{\downarrow} \equiv \hat{U}^{(\downarrow)} \hat\varrho \hat{U}^{(\downarrow)\dagger}$ is called the \emph{passive state} associated with $\hat\varrho$, and it is the state with lowest energy in the unitary orbit of $\hat\varrho$. If a state is passive, then no more energy can be extracted from it using unitary transformations; it has zero ergotropy. Thus, $\ergo \geq 0$ describes how much work can be discharged from the battery state.

Conversely, the state $\hat\varrho^{\uparrow} \equiv \hat{U}^{(\uparrow)} \hat\varrho \hat{U}^{(\uparrow)\dagger}$ is known as the \emph{active state} associated with $\hat\varrho$ \cite{Salvia2021,Biswas2022,Tirone2022}, and it is the state with the highest energy in the unitary orbit of $\hat\varrho$. A state $\hat\varrho^{\uparrow}$ is active if and only if no further energy can be injected into it by means of unitary evolution; it has zero antiergotropy. Hence, $\antiergo \leq 0$, and its magnitude quantifies by how much the battery state can be charged.

Letting $\lambda_0 \leq\lambda_1\leq  \ldots \leq \lambda_{d-1}$ denote the eigenvalues of the quantum state $\hat\varrho$, and $\epsilon_0 \leq \epsilon_1\leq  \ldots \leq \epsilon_{d-1}$ the eigenenergies of the Hamiltonian $\ham=\sum_{i}\epsilon_i|\epsilon_i\rangle\langle \epsilon_i|$, the energy content of the extremal states becomes
\begin{eqnarray}
\Tr[\hat\varrho^{\downarrow} \ham] &=& \sum_{i=0}^{d-1} \lambda_i \epsilon_{d-1-i}  \; ;
\label{energia_pass} \\
\Tr[\hat\varrho^{\uparrow} \ham] &=& \sum_{i=0}^{d-1} \lambda_i \epsilon_i   \; .
\label{energia_att}
\end{eqnarray}
Accordingly, the ergotropy (antiergotropy) is obtained by subtracting the energy content of the passive (active) state from the initial mean energy $\Tr [\hat\varrho\ham]$.

The ergotropy is a sublinear and convex functional \cite{Perarnau2015},  given the Hamiltonian $\ham=\sum_{i}\epsilon_i|\epsilon_i\rangle\langle \epsilon_i|$ with $\epsilon_i\leq\epsilon_{i+1}$, for any $\hat\varrho$ and $\hat\tau$ such that $\Tr[\hat\varrho\ham ]=\Tr[\hat\tau\ham ]$, then we have
\begin{eqnarray}
\ergo(t\hat\varrho+(1-t)\hat\tau; \ham)\leq t\ergo(\hat\varrho; \ham)+(1-t)\ergo(\hat\tau; \ham) \;.
\label{Ergoconvex}
\end{eqnarray}
Using~(\ref{Ergoconvex}) and the identity $\antiergo(\hat\rho; \ham) = -\ergo(\hat\rho; -\ham)$ it is also immediate to see that the opposite holds for the antiergotropy,
\begin{eqnarray}
\antiergo(t\hat\varrho+(1-t)\hat\tau; \ham)\geq t\antiergo(\hat\varrho; \ham)+(1-t)\antiergo(\hat\tau; \ham) \;.
\label{Antiergoconcave}
\end{eqnarray}

\textit{The quantum battery capacity.}---Both the ergotropy and antiergotropy of a quantum system are not constant during an isentropic thermodynamic cycle. However, their difference is constant during any unitary evolution. Here we call it the \emph{battery capacity} of the system.

\begin{definition}
The battery capacity of a quantum state $\hat\varrho$ with respect to a Hamiltonian $\ham$ is given by
\begin{eqnarray}
\workrange(\hat\varrho; \ham) = \ergo(\hat\varrho; \ham) - \antiergo(\hat\varrho; \ham) = \Tr[\hat\varrho^\uparrow \ham] - \Tr[\hat\varrho^\downarrow \ham] \; .
\label{def_workcapacity}
\end{eqnarray}
\end{definition}

The battery capacity represents the amount of work that a quantum system can transfer during any thermodynamic cycle that keeps the battery's evolution unitary (as is the case for a quantum battery which is thermally isolated, but mechanically coupled to work source or a load). We can write $\workrange(\hat\varrho; \ham)$ as the difference between the energies of the two extremal states in the unitary orbit of $\hat\varrho$: the active state $\hat\varrho^{\uparrow}$, which realizes the maximum possible energy~\eqref{energia_att}, and the passive state $\hat\varrho^{\downarrow}$, with energy~\eqref{energia_pass}; see Fig.~\ref{inject}. Equivalently, the work capacity of a state $\hat\varrho$ is equal to the ergotropy of the active state associated with $\hat\varrho$ minus the antiergotropy of the relative passive state.

It is apparent from the definition that $\workrange(\hat\varrho; \ham)$ is an unitarily invariant functional of the state, i.e, $\workrange(\hat\varrho; \ham) = \workrange(\hat{U} \hat\varrho \hat{U}^\dagger; \ham) $. The battery capacity thus admits a simple expression in terms of the eigenvalues $\{ \lambda_i \}$ of the density matrix and of the energy levels $\{ \epsilon_i \}$ of the Hamiltonian. From~(\ref{energia_pass}),~(\ref{energia_att}), and from the definition~(\ref{def_workcapacity}), we deduce
\begin{eqnarray}
\workrange(\hat\varrho; \ham) &=& \sum_{i=0}^{d-1} \epsilon_i \left( \lambda_i - \lambda_{d-1-i} \right) \nonumber
\\
&=& \sum_{i=0}^{d-1} \lambda_i \left( \epsilon_{i} -\epsilon_{d-i-1}\right) \; .
\end{eqnarray}
Moreover, from~(\ref{def_workcapacity}),~(\ref{Ergoconvex}), and~(\ref{Antiergoconcave}) the battery capacity is, like the ergotropy, a convex and sublinear functional,
\begin{equation}
 \workrange(t\hat\varrho+(1-t)\hat\tau; \ham)\leq t\workrange(\hat\varrho; \ham)+(1-t)\workrange(\hat\tau; \ham).
\end{equation}
Finally, the invariance with respect to unitary transformation results in the following property:

\begin{prop}\label{majorization}
The battery capacity is a Schur-convex functional of $\hat\varrho$. That is, if a state $\hat\varrho$ is majorized by $\hat\tau$, ($\hat\varrho\prec\hat\tau$), then $\workrange(\hat\varrho; \ham) \leq \workrange(\hat\tau; \ham)$.
\end{prop}
\begin{proof}

From Ref.\cite{Alimuddin2020b} the passive-state energy $\Tr[\hat\varrho^\downarrow\ham]$ is a Schur-concave functional of $\hat\varrho$, implying that $-\Tr[\hat\varrho^\downarrow\ham]$ is Schur-convex \cite{Marshall2010}. Given that the passive state of $\hat\varrho$ with respect to the Hamiltonian $\ham$ is the active state of $\hat\varrho$ with respect to $-\ham$, the energy of the active state $\Tr[\hat\varrho^\uparrow\ham]$ is Schur-convex. Therefore the battery capacity, as the sum of two Schur-convex functionals, is Schur-convex.
\end{proof}

An important lower bound on the battery capacity can be given in terms of the state purity and the spread of the battery's energy spectrum, as measured by the variance of the Hamiltonian, $\sigma^2_{\ham} = \Tr[\ham^2] - \Tr[\ham]/d$\cite{Salvia2021}.

\begin{prop}
Letting $\sigma_{\ham} = \sqrt{\sigma^2_{\ham}}$ and $\sigma_{\hat\varrho} = \sqrt{\Tr[\rho^2] - 1/d}$, the capacity of a $d$-dimensional battery is bounded by
\begin{eqnarray}
\workrange(\hat\varrho; \ham) \geq 2 \frac{\sigma_{\ham} \sigma_{\hat\varrho}}{\sqrt{d^2-1}} \; .
\label{bound_withvariance}
\end{eqnarray}
\end{prop}
See Appendix A  for the proof.

\textit{Battery capacity of many battery copies.}---It is possible to extract more work from, or charge more work to, an ensemble of $n$ identical copies of a quantum battery by using global operations on the whole ensemble \cite{Alicki2013,Campaioli2018,Safranek2022}. The figure of merit that quantifies the maximum work that can be extracted in this regime is the \emph{total ergotropy}, defined as~\cite{Tirone2022}:
\begin{eqnarray}
\ergo_{\rm tot}( \hat\varrho ; \ham ) \equiv \lim_{n \to \infty} \frac{1}{n} \ergo\left( \hat\varrho^{\otimes n} ;\ham^{(n)} \right) \; ,
\end{eqnarray}
with $\ham^{(n)}$ the Hamiltonian of $n$ non-interacting copies of the system. The total ergotropy can also be expressed as
\begin{equation}
 \ergo_{\rm tot}( \hat\varrho ; \ham ) = \Tr[\hat\varrho\ham] - \Tr[\hat\omega_{\beta(\hat\varrho)} \ham],
\end{equation}
where $\hat\omega_{\beta(\hat\varrho)} = e^{-\beta(\hat\varrho) \ham} / \mathcal{Z}$ is the Gibbs state of thermal equilibrium with a unique inverse temperature $\beta(\hat\varrho)$ such that the von Neumann entropies match, $S(\hat\omega_{\beta(\hat\varrho)}) = S(\hat\varrho)$. (Equivalently, the Gibbs state is the state of lowest energy among those with the same entropy as $\hat\varrho$.)

Similarly, we can define and express the \emph{total antiergotropy} as
\begin{eqnarray}
\antiergo_{\rm tot}( \hat\varrho ; \ham ) &\equiv& \lim_{n \to \infty} \frac{1}{n} \antiergo\left( \hat\varrho^{\otimes n}; \ham^{(n)} \right) \\
&=& \Tr[\hat\varrho\ham] - \Tr[\hat\omega_{-\beta^\ast(\hat\varrho)}; \ham]\; , \nonumber
\end{eqnarray}
where $\hat\omega_{-\beta^\ast(\hat\varrho)}$ is the inverse Gibbs state with negative inverse temperature such that $S(\hat\omega_{-\beta^\ast(\hat\varrho)}) = S(\hat\varrho)$, which is also the state of highest energy among those with the same entropy as $\hat\varrho$.

The battery capacity of an ensemble of $n \gg 1$ identical quantum systems will tend to the ``entropy-dependent battery capacity'' defined in Ref.\cite{Salamon2020},
\begin{eqnarray}
\workrange_{\rm tot}( \hat\varrho ; \ham ) &=& \lim_{n \to \infty} \frac{1}{n} \workrange\left( \hat\varrho^{\otimes n} ; \ham^{(n)} \right)  \nonumber \\
&=&\Tr[\hat\omega_{-\beta^\ast(\hat\varrho)}; \ham] - \Tr[\hat\omega_{\beta(\hat\varrho)}; \ham] \; .
\end{eqnarray}

\textit{The capacity of a two-level battery.}---%
We now consider the simplest example of a battery: a quantum system made of two levels $\ket{0}$ and $\ket{1}$, with corresponding Hamiltonian $\ham = E\ket{1}\bra{1}$. The battery capacity can be related to entropic quantities and measures of coherence \cite{Zheng,Napoli,Franc2020}. We generalize our findings to a $d$-dimensional battery with equally spaced energy levels in Section B of the supplemental material.

The density matrix $\hat{\varrho}$ on a two-level system can be written as
\begin{eqnarray}
\hat \varrho&=&
\begin{bmatrix}
1-q & ce^{i\theta} \\
ce^{-i\theta} & q
\end{bmatrix}
 \label{state} \; ;
\end{eqnarray}
with $q\in [0, 1]$ the population of the excited state, $c\in [0, \sqrt{q(1-q)} ]$ the amount of coherence in the state, and $\theta\in[0,2\pi]$. Herein, the governing Hamiltonian  is $\ham=E|1\rangle\langle1|$. The two eigenvalues of the density matrix~(\ref{state}) are $\lambda_\pm = [1 \pm  \sqrt{(2q-1)^2 + 4c^2}]/2$, with $\lambda_+ \geq \max \{ q, 1-q \}$ and $\lambda_- \leq \min \{ q,1-q \}$. The ergotropy of this state is $\ergo(\hat{\varrho}; \ham) = E (q - \lambda_-)$, while its antiergotropy is $\antiergo(\hat{\varrho}; \ham) = E ( q-\lambda_+)$. Hence, the battery capacity is
\begin{eqnarray}
\workrange(\hat\varrho; \ham)=E(1-2\lambda_-)=E\sqrt{(2q-1)^2 + 4c^2} \; .
\label{workrange_qubit}
\end{eqnarray}
This simple quantum battery is represented graphically in Fig.~\ref{inject}.

We observe that the base-2 von Neumann entropy $S(\hat\varrho)=-\Tr(\hat\varrho\log_2\hat\varrho)$ and the capacity of a two-level battery satisfy the inequality:
\begin{eqnarray}
\frac{\workrange(\hat\varrho; \ham)}{E}+S(\hat\varrho)\geq 1,
\label{capacity-Sentropy}
\end{eqnarray}
with equality only for pure states or the completely mixed state. This follows by virtue of (\ref{workrange_qubit}) with the inequality $S(\hat\varrho) \geq 2\lambda_- $ for $\lambda_-\in [0, 1/2]$.

A similar inequality, but in the opposite direction, holds for a range of Tsallis entropies, defined by \cite{Tsallis}:
\begin{equation}
 T_p(\hat\varrho)=\frac{1-\Tr \hat\varrho^p}{p-1} = \frac{1-\lambda_-^p-(1-\lambda_-)^p}{p-1} \; .
 \label{eq:tsallisentropy}
\end{equation}
For orders $p \geq 2$, we find
\begin{eqnarray}
\frac{\workrange(\hat\varrho; \ham)}{E} + T_p(\hat\varrho)\leq 1 \; .
\label{capacity-Tentropy}
\end{eqnarray}
This can be proven by using the function $g_p(\hat\varrho) = 2\lambda_- - T_p(\hat\varrho)$, which is monotonically increasing in $\lambda_- \in [0,1/2]$ whenever $p\geq 2$.

Finally, for the special case of the linear entropy, $L(\hat\varrho)\equiv T_2(\hat\varrho)=1-{\rm Tr}[\hat\varrho^2]=1-\lambda^2_- - \lambda^2_+$ \cite{HHH}, one easily obtains the equality:
\begin{eqnarray}
\frac{\workrange^2(\hat\varrho; \ham)}{E^2} + 2L(\hat\varrho)=1 \; .
\label{linearTentropy}
\end{eqnarray}
We prove similar operational relationships for equidistant $d$-level batteries in Appendix B.

We now turn to the relations between capacity and coherence. Three of the most common measures of coherence for quantum states are: the $l_1$-norm of coherence measuring the overall magnitude of off-diagonal elements, $\mathsf{Cohe}_{l_1}(\hat\varrho)=\sum_{i\neq j}\lvert\varrho_{i,j}\rvert$; the robustness of coherence  \cite{Napoli,Zheng},
\begin{equation}
 \textsf{Cohe}_{\rm RoC}(\hat\varrho)=\min_{\hat\tau \in \mathcal{D}(\mathbbm{C}^d)} \left\{s\geq 0 \, \bigg| \, \frac{\hat\varrho+s\hat\tau}{1+s} \in \mathcal{F} \right\},
\end{equation}
with $\mathcal{D}(\mathbbm{C}^d)$ the convex set of $d$-dimensional density operators and $\mathcal{F} \subset \mathcal{D}(\mathbbm{C}^d)$ the subset of incoherent states; and the relative entropy of coherence \cite{Baumgratz}, $\textsf{Cohe}_{\rm re}(\hat\varrho) = S(\hat\varrho_{\rm {inc}})-S(\hat\varrho)$, where $\hat\varrho_{\rm {inc}}$ is the state obtained by deleting all the off-diagonal elements from $\hat\varrho$.

For qubit states~\eqref{state}, the first two measures are equivalent,
\begin{eqnarray}
\mathsf{Cohe}_{l_1}(\hat\varrho)= \textsf{Cohe}_{\rm RoC}(\hat\varrho) = 2c \; .
\label{coerenza_qubit}
\end{eqnarray}
Hence the capacity \eqref{workrange_qubit} of a qubit battery can be decomposed into an incoherent and a coherent part,
\begin{eqnarray}
\workrange^2(\hat\varrho; \ham) &=& \workrange^2(\hat\varrho_{\rm inc}; \ham) + E^2 \mathsf{Cohe}^2_{l_1}(\hat\varrho) \nonumber \\
&=& \workrange^2(\hat\varrho_{\rm inc}; \ham) + E^2 \mathsf{Cohe}^2_{\rm RoC}(\hat\varrho) \; ,
\end{eqnarray}
where the incoherent part, $\workrange(\hat\varrho_{\rm inc}; \ham) = (1-2q)E$ for $q\in[0, \frac{1}{2}]$, is the battery capacity of the diagonal state $\hat\varrho_{\rm inc}$.

A similar decomposition does not hold for the relative entropy of coherence; however, a simple substitution from~\eqref{capacity-Sentropy} yields the inequality
\begin{eqnarray}
1+\textsf{Cohe}_{\rm re}(\hat\varrho) \leq \frac{\workrange(\hat\varrho; \ham)}{E} +S(\hat\varrho_{\rm inc}).
\end{eqnarray}
General cases of $d$-dimensional batteries are shown in Section B.

\textit{The capacity gap as an entanglement measure.}---In the case of composite quantum batteries comprised of two or more local Hamiltonians, an entangled battery state can accommodate non-local work storage that is more than the sum of its local parts. This gives rise to energy-based entanglement criteria for bipartite and multipartite systems.

Consider first a bipartite state $\hat\varrho$ on the Hilbert space ${\cal H}_A\otimes \cal {\cal H}_B$, with Hamiltonian $\ham={\hat H}_A\otimes \mathbbm{I}_B+\mathbbm{I}_A\otimes {\hat H}_B$. The \emph{ergotopic gap} $\delta_{\rm out}$ is the difference of the ergotropy obtained by global unitary operations and local unitary operations:
\begin{eqnarray}
\delta_{\rm out}(\hat\varrho; \ham)\equiv \ergo(\hat\varrho; \ham) - \ergo_{\rm L}(\hat\varrho; \ham)
 \nonumber \\
 = \max_{\hat{U} \in \mathbf{U}(d^2)} W_{\hat U}(\hat\varrho; \ham) - \max_{\hat{U}_\ell \in \mathbf{U}_{\rm L}(d^2)} W_{\hat U_\ell}(\hat\varrho; \ham) \; ,
\end{eqnarray}
where $\mathbf{U}_{\rm L}(d^2)$ is the group of local unitary operations of the form $\hat{U}_\ell = \hat{U}_A \otimes \hat{U}_B$.
Similarly, we can define the difference of the antiergotropy as
\begin{eqnarray}
\delta_{\rm in}(\hat\varrho; \ham) \equiv  \antiergo_{\rm L}(\hat\varrho; \ham) - \antiergo(\hat\varrho; \ham)
\nonumber \\
 = \min_{\hat{U}_\ell \in \mathbf{U}_{\rm L}(d^2)} W_{\hat U_\ell}(\hat\varrho; \ham)  - \min_{\hat{U} \in \mathbf{U}(d^2)} W_{\hat U}(\hat\varrho; \ham) \; .
\end{eqnarray}
The sum of $\delta_{\rm in}$ and $\delta_{\rm out}$ corresponds to the difference between the global capacity of the battery state and the battery capacity restricted to local operations. The latter is the sum of the individual capacities of the reduced battery states. We call the difference in global and local capacities the \textit{bipartite battery capacity gap}:
\begin{eqnarray} \label{delta_capacity}
\Delta_{A|B}(\hat\varrho; \ham) &\equiv& \delta_{\rm in}(\hat\varrho; \ham)+\delta_{\rm out}(\hat\varrho; \ham) \\
&=& \workrange(\hat\varrho; \ham) - \workrange(\hat\varrho_A; \ham_A) - \workrange(\hat\varrho_B; \ham_B) \; . \nonumber
\end{eqnarray}
This definition naturally extends to multipartite systems: the \emph{fully separable capacity gap} of an $n$-partite battery state $\hat\varrho$ with Hamiltonian $\ham=\sum_{i}\ham_{A_i}\otimes \mathbbm{1}$ will be
\begin{eqnarray}
\Delta_{A_1|\cdots |A_n}( \hat\varrho ; \ham ) \equiv \workrange( \hat\varrho ; \ham ) - \sum_{i=1}^n \workrange( \hat\varrho_{A_i}; \ham_{A_i} ) \; .
\end{eqnarray}

\begin{prop}\label{th1}
The fully separable battery capacity gap $\Delta_{A_1|\cdots |A_n}$ of a pure state $|\Psi\rangle$ on Hilbert space $\mathcal{H}=\otimes_{i=1}^n\mathcal{H}_{A_i}$, is non-increasing under local operations and classical communications (LOCC).
\end{prop}

Thanks to Proposition \ref{th1} (proven in Appendix C, the battery capacity gap can serve as a witness of entanglement in bipartite or multipartite systems. In Appendix C, we propose measures of genuine multipartite entanglement and give some elementary examples for bipartite and tripartite states.

\begin{figure}
\begin{center}
\includegraphics[width=\columnwidth]{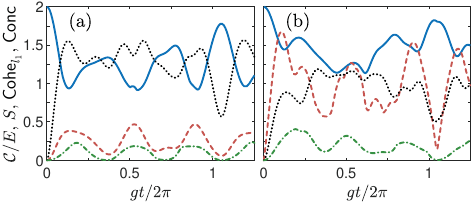}
\caption{Two-atom battery capacity as a function of interaction time with a resonant single-mode field in (a) a thermal state and (b) a coherent state of mean photon number $n_0=0.5$. We compare the capacity in units $E$ (${\workrange}/{E}$, solid line) to the base-2 von Neumann entropy ($S$, dotted), the $l_1$-coherence ($\mathsf{Cohe}_{l_1}$, dashed), and the concurrence ($\mathsf{Conc}$, dash-dotted). The battery starts in the maximum-capacity state $|eg\rangle$.}
\label{fig2}
\end{center}
\end{figure}

As a physical example, we consider a battery comprised of two two-level atoms interacting with a resonant, thermally or coherently populated cavity mode. The mode could serve to charge, discharge, or readout the battery state, but it also acts as a source of entropy and mediates coherence and entanglement, all affecting the battery capacity over time \cite{Andolina2019,Kim,Restrepo}. Given a uniform coupling rate $g$, we model this by a Tavis-Cummings Hamiltonian in the rotating wave approximation,
\begin{equation}
\ham=E \left( \hat{a}^\dagger\hat{a} + \frac{\hat\sigma_1^z + \hat\sigma_2^z}{2} \right)+ \hbar g\sum_{i=1}^2 \left( \hat{a}\hat\sigma_i^+ + h.c. \right),
\end{equation}
with $\hat a$ the cavity ladder operator, $\hat \sigma_i^z$ the Pauli-$z$ matrix of the $i$-th atom, and $\hat\sigma_i^+ = |e\rangle_i\langle g|$ the $i$-th excitation operator. Starting from a maximum-capacity state of the battery and a thermal cavity state with mean population $n_0$, $\hat\rho(0) = |eg\rangle\langle eg| \otimes \sum_{n=0}^\infty |n\rangle\langle n| n_0^n/(n_0+1)^{n+1}$, the reduced battery state $\hat\varrho(t) = \sum_n \langle n|\hat\rho(t)|n\rangle$ evolves as a mixed, symmetric two-qubit state over time $t$ \cite{Kim,Restrepo}, with varying amounts of coherence and atom-atom entanglement.

In Fig.~\ref{fig2}, we compare the two-atom battery capacity in $E$-units, $\workrange (\hat\varrho(t),(\hat\sigma_1^z + \hat\sigma_2^z)/2)$, against the von Neumann entropy, the $l_1$-coherence, and the atom-atom entanglement in terms of the concurrence \cite{Wootters}, for $n_0=0.5$. The capacity for an initially thermal cavity state in (a) drops to lower values and oscillates more strongly than the case of a coherent state in (b). Both examples illustrate that, while the rise in entropy as well as the transient oscillations of coherence and entanglement clearly influence the capacity, neither of those quantities alone can adequately predict its behaviour. See Appendix D for more examples at other temperatures, couplings, and for other squeezed cavity states.

\textit{Conclusions.}---We have introduced the capacity of a quantum battery system as the difference between the maximal and the minimal energy that can be reached from it by unitary evolution. It quantifies the amount of work that a quantum battery can at most supply during operation cycles. The battery capacity does not depend on the actual battery charge at any given moment, making it a suitable figure of merit for comparing different quantum battery models.

Due to its unitary invariance, the battery capacity can be put in relation with the entropy of the battery state, and with measures of coherence and entanglement, as we have discussed for simple models with an equidistant energy level spectrum. We hope that extending this analysis to other quantum battery models will lead to deeper insights into the connection between quantum thermodynamics, work storage, and quantum information theory. Future works could explore similar figures of merit for the capacity of quantum systems to store other resources.

\acknowledgments
We are grateful to Qinghu Chen and Miao Zhang for dicussions. This work was supported by the National Natural Science Foundation of China (Grants No. 62172341, No. 12204386, No. 12075159, and No. 12171044), National Natural Science Foundation of Sichuan Provence (No. 23NSFSC0752), Beijing Natural Science Foundation (No. Z190005), the Academician Innovation Platform of Hainan Province and Shenzhen Institute for Quantum Science and Engineering, Southern University of Science and Technology (No. SIQSE202105), and I-HUB Quantum Technology Foundation (No. I-HUB/PDF/2021-22/008).


\appendix
\renewcommand{\appendixname}{Appendix~}
\renewcommand{\thesection}{\Alph{section}}
\renewcommand{\thesubsection}{\arabic{subsection}}
\renewcommand{\theequation}{S\arabic{equation}}
\renewcommand{\thefigure}{S\arabic{figure}}

\section{Proof of the inequality~(12)}

\label{sec:bound_varianza}
To prove the inequality ~(12), consider the probability distribution $P(E ; \hat\varrho; \ham)$ of the work extracted with a random unitary transformation; i.e., the probability distribution of $W_{\hat{U}}(\hat\varrho; \ham)$ when $\hat{U}$ is sampled uniformly with respect to the Haar measure. In Ref.\cite{Salvia2021} it has been proven that the variance of $P(E ; \hat\varrho; \ham)$ is given by
\begin{eqnarray}
\underset{{\hat{U} \sim \mbox{Haar}(\textbf{U}(d)) }}{\mbox{Var}} \left[ W_{\hat{U}}(\hat\varrho; \ham) \right] = \frac{\sigma^2_{\hat\varrho}\sigma^2_{\ham}}{d^2 - 1}.
\label{varianza_haar}
\end{eqnarray}
We invoke Popoviciu's inequality on variance:
\begin{eqnarray}
\underset{{x \sim P(x)}}{\mbox{Var}} \left[ f(x) \right] \leq \frac{1}{4}\left( \max_{x \sim P(x)} f(x) - \min_{x \sim P(x)} f(x) \right)^2 \;
\label{Popoviciu}
\end{eqnarray}
with the identifications $x = \hat{U}$, $P(x) = \mbox{Haar}(\textbf{U}(d))$, and $f(\hat{U}) = W_{\hat{U}}(\hat\varrho; \ham)$. Combined with~(2) and~(9), this implies that
\begin{eqnarray}
\underset{{\hat{U} \sim \mbox{Haar}(\textbf{U}(d)) }}{\mbox{Var}} \left[ W_{\hat{U}}(\hat\varrho; \ham) \right] \leq \frac{1}{4}\workrange^2(\hat\varrho; \ham) \; ,
\end{eqnarray}
which combined with~(\ref{varianza_haar}) proves the inequality (12) in the main text.

\section{Battery capacity for equally spaced energy levels}
\label{sec:equispaced}

Here we provide bounds for the battery capacity as well as relations to entropy and coherence functionals for a $d$-level battery with an equally spaced spectrum, $\ham=\sum^{d-1}_{j=0} j E \ket{j}\bra{j}$.

\begin{prop}\label{Dualityproperty0}
Given a state $\varrho$ with $d$ eigenvalues $\{\lambda_1,\cdots,\lambda_{d-1}\}$ arranged in an increasing order and Hamiltonian $\ham=\sum^{d-1}_{j=1}jE|j\rangle\langle j|$, we have
\begin{eqnarray}
\Tr (\hat\varrho^\uparrow \ham)+ \Tr (\hat\varrho^\downarrow \ham)=(d-1)E.
\label{energygap1}
\end{eqnarray}
\end{prop}

\begin{proof}
The energies of the states $\varrho^\downarrow$ and $\varrho^\uparrow$ are given by
\begin{eqnarray}
&&\Tr(\hat\varrho^\downarrow \ham)=\sum^{d-1}_{i=0}i\lambda_{d-1-i}E,
\label{framartinocampanaro}
\end{eqnarray}
and
\begin{eqnarray}
&&\Tr(\hat\varrho^\uparrow \ham)=\sum^{d-1}_{i=0}i\lambda_iE.
\label{dormitudormitu}
\end{eqnarray}
Summing~(\ref{framartinocampanaro}) and~(\ref{dormitudormitu}) we obtain
\begin{equation}
\Tr (\hat\varrho^\uparrow \ham)+ \Tr (\hat\varrho^\downarrow \ham)=dE-1\sum^{d-1}_{i=0}\lambda_iE=(d-1)E.
\label{Duality1}
\end{equation}
\end{proof}

We also remark that, in a Hamiltonian with equispaced energy levels, the ergotropic and antiergotopic gap coincide:

\begin{prop}
Given the HamiltonianS $\ham_A=\sum^{d_A-1}_{j=1}jE|j\rangle\langle j|$ and $\ham_B=\sum^{d_B-1}_{j=1}jE|j\rangle\langle j|$, for any state $\hat\varrho$ it holds
\begin{eqnarray}
\delta_{\rm in}(\hat\varrho_{AB}; \ham_{AB}) = \delta_{\rm out}(\hat\varrho_{AB}; \ham_{AB}).
\end{eqnarray}
\end{prop}

\begin{proof}
It follows from the fact that this Hamiltonian satisfies $\ham = - \ham + C\mathbbm{1}$, where $C=(d-1)E$ is the value of the largest energy level. Therefore $\delta_{\rm in}(\hat\varrho; \ham) = \delta_{\rm out}(\hat\varrho; -\ham) = \delta_{\rm out}(\hat\varrho; \ham - C) = \delta_{\rm out}(\hat\varrho; \ham)$, where in the last passage we have used the invariance of the ergotropic gap with respect to shifts of the system Hamiltonian.
\end{proof}

\textit{Example S1}. Consider the system Hamiltonian $\ham_A=\ham_B=E|1\rangle\langle1|$ and the family of Werner states \cite{Werner} given by
\begin{eqnarray}
\hat\varrho_{v}=v|\psi\rangle\langle\psi|+\frac{1-v}{4}\mathbbm{1},
\end{eqnarray}
where $|\psi\rangle=\cos\theta|00\rangle+\sin\theta|11\rangle$, with $\theta\in (0,\pi/4]$ and $v\in [0,1]$. The spectral values are $\{\frac{1+3v}{4}, \frac{1-v}{4},\frac{1-v}{4},\frac{1-v}{4}\}$, for this entire class of states, its reduced density matrix has the spectra $\{v\cos^2\theta+\frac{1-v}{2}, v\sin^2\theta+\frac{1-v}{2}\}$,  it turns out to be
\begin{eqnarray}
\delta_{\rm in}(\hat\varrho_{v})=\delta_{\rm out}(\hat\varrho_{v})=2v \sin^2\theta E.
\end{eqnarray}

It is easy to bound the quantum battery capacity of this system with
\begin{eqnarray}\label{D1}
\mathcal{C}(\hat\varrho)&=&\sum_{j=0}^{d-1}jE(\lambda_j-\lambda_{d-1-j})
\\
&=&\sum_{j=0}^{\lfloor d/2 \rfloor}E(d-1-2j)(\lambda_{d-1-j}-\lambda_j)
\nonumber \\
&\leq &\sum_{j=0}^{\lfloor d/2 \rfloor}E(d-1-2j)(\lambda_{d-1}-\lambda_0), \nonumber
\end{eqnarray}
and similarly
\begin{eqnarray}\label{D2}
\mathcal{C}(\hat\varrho)&=&\sum_{j=0}^{d-1}jE(\lambda_j-\lambda_{d-1-j})
\\
&=&\sum_{j=0}^{\lfloor d/2 \rfloor}E(d-1-2j)(\lambda_{d-1-j}-\lambda_j)
\nonumber
\\
&\geq &\sum_{j=0}^{\lfloor d/2 \rfloor}E(d-1-2j)(\lambda_{\lfloor d/2 \rfloor+1}
-\lambda_{\lfloor d/2 \rfloor}).
\nonumber
\end{eqnarray}
From~(\ref{D1}) and~(\ref{D2}) we have the following simple bounds for the battery capacity in a system with equispaced energy levels:
\begin{eqnarray}
\left( \left\lfloor \tfrac{d}{2} \right\rfloor \right)^2 \left( \lambda_{d-1} - \lambda_{0} \right) &\geq &\frac{\workrange(\hat\varrho; \ham)}{E}
\nonumber
\\
&\geq& \left( \left\lfloor \tfrac{d}{2} \right\rfloor \right)^2 \left( \lambda_{\lfloor d/2 \rfloor + 1} - \lambda_{\lfloor d/2 \rfloor} \right). \nonumber \\
\end{eqnarray}

\subsection{Relationships between battery capacity and entropic functionals}

In this section we shall use the inequality~(12) to derive an inequality between battery capacity and entropy for an equispaced Hamiltonian $\ham=\sum_{j=1}^{d-1}jE|j\rangle\langle{j}|$.
The variance of an equispaced Hamiltonian is given by
\begin{eqnarray}
\sigma_{\ham}^2&=&\Tr\ham^2-\frac{\left(\Tr\ham\right)^2}{d}
\nonumber
\\
&=&\sum_{j=0}^{d-1}j^2E^2-\frac{E^2}{d}\left(\sum_{j=0}^{d-1}j\right)^2
\nonumber
\\
&=&\frac{(d^2-1)(2d-3)}{6}.
\label{HHD1}
\end{eqnarray}
Moreover, it is always true that
\begin{eqnarray}
\sigma_{\varrho}^2&=&\Tr\varrho^2-\frac{(\Tr\varrho)^2}{d}
\nonumber
\\
&=&\frac{d-1}{d}-L(\hat\varrho),
\label{HHD2}
\end{eqnarray}
where $L(\hat\varrho) = 1 - \Tr[\hat\varrho^2]$ is the linear entropy.

Replacing~(\ref{HHD1}) and~(\ref{HHD2}) into the inequality (12) we obtain the inequality
\begin{eqnarray}
\frac{\mathcal{C}^2(\hat\varrho)}{E^2}
\geq \frac{2(2d-3)(d-1)}{3d}-\frac{(4d-6)}{3} L(\hat\varrho),
\label{HHD7}
\end{eqnarray}
and therefore
\begin{eqnarray}
\frac{\mathcal{C}^2(\hat\varrho)}{E^2}+\frac{4d-6}{3} L(\hat\varrho)\geq \frac{2(2d-3)(d-1)}{3d}.
\label{HHD8}
\end{eqnarray}

Using (\ref{HHD8}) together with the inequailty $S(\hat\varrho)\geq L(\hat\varrho)$ we get a similar inequality for the Von Neumann entropy:
\begin{eqnarray}
\frac{\mathcal{C}^2(\hat\varrho)}{E^2}+\frac{4d-6}{3} S(\hat\varrho)\geq \frac{2(2d-3)(d-1)}{3d}.
\label{HHD9}
\end{eqnarray}

For the Tsallis entropy defined by $T_q(\hat\varrho)=\frac{1}{q-1}(1-\Tr\varrho^q)$ \cite{Tsallis}, it is easy to show that $T_q(\hat\varrho)\geq \frac{1}{q-1}L(\hat\varrho)$ for any $q\geq 2$. Therefore, from the inequality (\ref{HHD8}) we also get
\begin{eqnarray}
\frac{\mathcal{C}^2(\hat\varrho)}{E^2}+\frac{(4d-6)(q-1)}{3} T_q(\hat\varrho)
\nonumber
\\
\geq \frac{2(2d-3)(d-1)}{3d}.
\label{HHD10}
\end{eqnarray}

\subsection{Relationships with coherence}

In this subsection we explore the quantum battery capacity and quantum coherence in high-dimensional Hilbert spaces with equispaced energy levels.

From the definition of $l_1$-coherence we have
\begin{eqnarray}
\textsf{Cohe}_{l_1}(\hat\varrho) +1=\|\hat\varrho\|_1\geq \|\hat\varrho\|_\infty
\label{diseguaglianza_norme}
\end{eqnarray}
where $\|\hat\varrho\|_\infty=\max\{\lambda(\hat\varrho)\}$, is the maximum eigenvalue of $\hat\varrho$. From (\ref{diseguaglianza_norme}) and (\ref{D1}) follows that
\begin{eqnarray}
\frac{\mathcal{C}(\hat\varrho)}{E}\leq
\left( \left\lfloor \frac{d}{2} \right\rfloor\right)^2 (\textsf{Cohe}_{l_1}(\hat\varrho) +1) \; .
\label{HHD13a}
\end{eqnarray}

From the inequality (\ref{HHD9}) we also get a simple relationship for the relative entropy of coherence $\textsf{Cohe}_{\rm re}(\hat\varrho) = S(\hat\varrho_{\rm {ic}})-S(\hat\varrho)$, namely
\begin{eqnarray}
\frac{1}{E^2} \mathcal{C}^2(\hat\varrho)+\frac{4d-6}{3} S(\hat\varrho_{\rm {ic}})-\frac{4d-6}{3}\textsf{Cohe}_{\rm re}(\hat\varrho)
\nonumber
\\
\geq \frac{2(2d-3)(d-1)}{3d}.
\label{HHD12}
\end{eqnarray}

For the robustness of coherence (RoC) of a quantum state $\hat\varrho$ we use the following inequality, proven in Ref.\cite{Napo}:
\begin{eqnarray}
\textsf{Cohe}_{\rm roc}(\hat\varrho)\geq \|\hat\varrho-\hat\varrho_{\rm {ic}}\|_2^2\geq \|\hat\varrho\|^2-\|\hat{\varrho}_{\rm {ic}}\|_2^2.
\label{coheroc_inequality}
\end{eqnarray}
Using (\ref{coheroc_inequality}) into (\ref{D1}) we can see that
\begin{eqnarray}
\frac{\mathcal{C}(\hat\varrho)}{E}\leq
\left( \left\lfloor \frac{d}{2} \right\rfloor\right)^2(\textsf{Cohe}_{\rm roc}(\hat\varrho)+\|\hat{\varrho}_{\rm {ic}}\|_2^2) \; .
\label{HHD13}
\end{eqnarray}
Moreover, it has shown that RoC and $l_1$ norm coherence \cite{Napo} satisfy
\begin{eqnarray}
\frac{1}{d-1}\textsf{Cohe}_{l_1}(\hat\varrho) \leq \textsf{Cohe}_{\rm roc}(\hat\varrho)\leq \textsf{Cohe}_{l_1}(\hat\varrho).
\end{eqnarray}
Combining the above inequalities with (\ref{HHD13a}) we get
\begin{eqnarray}
\frac{\mathcal{C}(\hat\varrho)}{E}\leq
\left( \left\lfloor \frac{d}{2} \right\rfloor\right)^2 ((d-1)\textsf{Cohe}_{\rm roc}(\hat\varrho) +1) \; .
\label{HHD13b}
\end{eqnarray}

\textit{Example S2}. Consider the Hamiltonian $\ham=\sum_{j=1}^{d-1}jE|j\rangle\langle{j}|$ and the Werner state \cite{Werner}:
\begin{eqnarray}
\hat{\varrho}_v=\frac{1-v}{d}\mathbbm{1}+v|\phi\rangle\langle\phi|,
\label{HHD14}
\end{eqnarray}
where $|\phi\rangle=\frac{1}{\sqrt{d}}\sum_{i=0}^{d-1}|i\rangle$ and $v\in [0,1]$. The eigenvalues of \eqref{HHD14} are given by $\{\frac{1+(d-1)v}{d}, \frac{1-v}{d}, \cdots, \frac{1-v}{d}\}$. The quantum battery capacity is given by
\begin{eqnarray}
\mathcal{C}(\hat\varrho_v)=(d-1)vE
\label{HHD15}
\end{eqnarray}
from \eqref{D1}. The von Neumann entropy is given by
\begin{eqnarray}
S(\hat\varrho_v)&=&-\frac{1+(d-1)v}{d}(\log_2(1+(d-1)v)-\log_2d)
\nonumber
\\
&&-\frac{(d-1)(1-v)}{d}(\log_2(1-v)-\log_2d).
\label{HHD16}
\end{eqnarray}
For the linear entropy we get
\begin{eqnarray}
L(\hat\varrho_v)&=1-\frac{1}{d^2}((1+(d-1)v)^2
\nonumber
\\
&+(d-1)(1-v)^2),
\label{HHD17}
\end{eqnarray}
while the Tsallis entropy is
\begin{eqnarray}
T_p(\hat\varrho_v)=\frac{1}{p-1}(1-\frac{1}{d^p}((1+(d-1)v)^p
\nonumber
\\
+(d-1)(1-v)^p)).
\label{HHD18}
\end{eqnarray}
For the coherence the $l_1$ norm coherence we get
\begin{eqnarray}
\textsf{Cohe}_{l_1}(\hat\varrho)=v(d-1).
\end{eqnarray}
For the relative entropy of coherence we have
\begin{eqnarray}
\textsf{Cohe}_{\rm re}(\hat\varrho)= \log_2d-S(\hat\varrho).
\end{eqnarray}
The robustness of coherence is given by
\begin{eqnarray}
\textsf{Cohe}_{\rm roc}(\hat\varrho)\leq \textsf{Cohe}_{l_1}(\hat\varrho)=v(d-1).
\end{eqnarray}
All of these bounds are shown in Table \ref{tables1}.
\begin{center}
\begin{table}
\caption{The quantum quantities of Werner state (\ref{HHD14}). WC denotes the quantum battery capacity. VE denotes the von Neumann entropy. LE denotes the linear entropy. TE denotes the Tsallis entropy. L1C denotes the $l_1$ norm coherence. ROC denotes the robustness of coherence.}
\begin{tabular}{|c|c|}
  \hline
  Quantities &  relationships\\   \hline
  WC &  =$E\times$ L1C \\   \hline
  VE & $\geq $ LE,TE \\   \hline
  LE &  = TE($q=2$) \\   \hline
  TE &  $\leq$  VE \\   \hline
  L1C & =WC $/E$  \\
  \hline
  ROC &  WC $/E$\\
  \hline
\end{tabular}
\label{tables1}
\end{table}
\end{center}

\section{Entanglement measures for multipartite entanglement}
\label{app:entanglement_measures}

Consider an $n$-partite pure quantum state $|\Psi\rangle$ on Hilbert space $\mathcal{H}=\otimes_{i=1}^n\mathcal{H}_{A_i}$. The Hamiltonian for the $i$-th subsystem is given by $\ham_{A_i}=\sum_{j=0}^{d_i-1}jE|j\rangle\langle j|$. This means the involved systems are not completely degenerate, i.e., that there are eigenstates with different energy. Without loss of generality, we associate zero energy to the lowest energetic state $\ket{0}$. The total interaction-free global Hamiltonian is given by $\ham=\sum_{i=1}^{n}\tilde{H}_{A_i}$, where $\tilde{H}_{A_i}=\mathbbm{1}_{d_1\cdots d_{i-1}}\otimes \ham_{A_i}\otimes\mathbbm{1}_{d_{i+1}\cdots d_n}$. The energy of a global state $\hat\varrho$ on $\mathcal{H}$ is given by $\Tr(\hat\varrho \ham)$, and the energy of a subsystem $A_i$ in the state $\hat\varrho_{A_i}=\Tr_{\forall A_j, j\not=i}(\hat\varrho)$ is given by $\Tr (\hat\varrho_{A_i}\ham_{A_i})$.

First we prove Proposition 3 given in the main text.

\textit{Proof of the Proposition 3}. This property follows from the fact that the fully separable ergotropic gap is non-increasing under LOCC \cite{Puliyil2022}, and that the energy of the active states, $-\Tr(\rho_{A_i}^\uparrow \ham_{A_i})$, is Schur-convex, see the proof of Proposition 1. Moreover, all the pure states have the same capacity of $\workrange (|\Psi\rangle\langle \Psi|,\ham)$ for a given Hamiltonian due to its unitary invariance.  $\Box$

In what follows, we extend bipartite entanglement measures to genuine multipartite entanglement measures.

From a pure state $\ket{\Psi}_{A_1\cdots A_n}$ both the maximum work extraction (ergotropy) and maximum work injection (antiergotropy) are possible by global unitary operations. Here, $\hat\varrho^\downarrow:=\ket{0}\bra{0}^{\otimes n}$ is the passive state and $\hat\varrho^\uparrow:=\ket{d_1-1,\cdots, d_n-1}\bra{d_1-1,\cdots, d_n-1}$ is the active state of $\ket{\Psi}\bra{\Psi}$. The maximum extractable work is $\ergo(\ket{\Psi}\bra{\Psi}; \ham)=\Tr(\ket{\Psi}\bra{\Psi} \ham)-\Tr (\hat\varrho^\downarrow \ham)$. Conversely, the maximum injectable work is $\antiergo(\ket{\Psi}\bra{\Psi}; \ham)=\Tr(\ket{\Psi}\bra{\Psi} \ham)-\Tr( \hat\varrho^\uparrow \ham)$.

On the other hand, the local unitary operations of $\otimes_{i=1}^nU_{A_i}$ cannot always turn a pure state $\ket{\Psi}\bra{\Psi}$ into the ground state of the system. For pure product states (i.e., $\ket{\Psi}\bra{\Psi} = \bigotimes _{i=1}^n \ket{\psi_i}\bra{\psi_i}$), the fully separable battery capacity gap is zero, where all the subsystems are pure and thus can be transformed into the lowest or highest energetic state by local unitary operations. But for pure biseparable states the fully separable capacity gap $\Delta(\ket{\Psi}\bra{\Psi}; \ham)$ may be not zero.

\textit{Example S3}. Consider the three-qubit pure state \cite{Acin2000},
\begin{eqnarray}
    \ket{\Psi}_{ABC}&=&\lambda_0\ket{000}+\lambda_1e^{i\theta}\ket{100}
    +\lambda_2\ket{101}
   \nonumber\\
    &&
    +\lambda_3\ket{110}+\lambda_4\ket{111},
    \label{e02}
\end{eqnarray}
where $\lambda_i\geq 0$, $\sum_i \lambda_i^2=1$ and $0\leq \theta\leq \pi$. Since the Hamiltonian for the marginal systems are $\ham_A=\ham_B=\ham_C=E\ket{1}\bra{1}$, passive state energy will be equal to the smallest eigenvalue while active state energy will be equal to the largest eigenvalue. In terms of the generalised Schmidt coefficients, we get the following energies as
\begin{eqnarray}
    \Tr(\hat\varrho^\uparrow \ham_A)&=&E-\frac{\delta_{\rm in}^{A|BC}}{2}
       \label{e066} \\
    &=&\frac{E}{2} \left(1+\sqrt{1-4\lambda_0^2(1-(\lambda_0^2+\lambda_1^2))} \right),
      \nonumber\\
    \Tr(\hat\varrho_B^\uparrow \ham_B)&=&E-\frac{\delta_{\rm in}^{B|CA}}{2}
  \nonumber\\
    &=& \frac{E}{2} \left( 1+\sqrt{1-4(\lambda_0^2(\lambda_3^2+\lambda_4^2)+\gamma)}
    \right),  \nonumber\\
    \Tr(\hat\varrho_C^\uparrow \ham_C)&=&E-\frac{\delta_{\rm in}^{C|AB}}{2}
  \nonumber\\
    &=& \frac{E}{2} \left(1+\sqrt{1-4(\lambda_0^2(\lambda_2^2+\lambda_4^2)+\gamma)}
    \right), \nonumber
\end{eqnarray}
where $\gamma=|(\lambda_1\lambda_4e^{i\theta}-\lambda_2\lambda_3)|^2$ and $\delta_{\rm in}^{X|YZ}$ is the biseparable ergotropic gap of the maximum charging work.

This implies the tripartite battery capacity gap by combining the ergotropic gap of extractable work \cite{Puliyil2022} as
\begin{eqnarray}
\Delta_{A|B|C}&=&\delta^{A|B|C}_{\rm in}+\delta^{A|B|C}_{\rm out}
 \nonumber
\\
&=& \frac{\delta_{\rm in}^{A|BC}+\delta_{\rm in}^{B|CA}+\delta_{\rm in}^{C|AB}}{2}
\nonumber
\\
&&+\frac{\delta_{\rm out}^{A|BC}+\delta_{\rm out}^{B|CA}+\delta_{\rm out}^{C|AB}}{2}
 \nonumber
\\
&=&
3E-E\sqrt{1-4\lambda_0^2(1-(\lambda_0^2+\lambda_1^2))}
 \nonumber
\\
&&-E\sqrt{1-4(\lambda_0^2(\lambda_3^2+\lambda_4^2)+\gamma)}
 \nonumber
\\
&&-E\sqrt{1-4(\lambda_0^2(\lambda_2^2+\lambda_4^2)+\gamma)} \nonumber
\\
&=&\frac{E}{2}(\Delta_{A|BC}+\Delta_{B|CA}+\Delta_{C|AB})
 \label{e077}
\end{eqnarray}
where $\Delta_{X|Y}=\delta_{\rm out}^{X|Y}+\delta_{\rm in}^{X|Y}$ denotes the battery capacity gap of biseparable system $X$ and $Y$.  Note that the state (\ref{e02}) is a pure product state for $\gamma=\lambda_0=0$ or $\lambda_2=\lambda_3=\lambda_4=0$. This follows $\delta_{\rm in}^{A|B|C}=0$. On the other hand, the quantity is optimal for the maximally entangled GHZ state \cite{GHZ} ($\lambda_0=\lambda_4=\frac{1}{\sqrt{2}}$) with the value $\Delta_{A|B|C}=3E$.

The fully separable ergotropic gap or battery capacity gap cannot be used to characterize the genuineness of multipartite entanglement because it might be nonzero for biseparable states. In what follows, we present some entanglement measures inspired by the ergotropic gap of extractable work \cite{Puliyil2022}.

\begin{definition}
The minimum of the biseparable battery capacity gap (MBWCG) of pure state $\ket{\Psi}$ on Hilbert space $\mathcal{H}=\otimes_{i=1}^n\mathcal{H}_{A_i}$ is defined by
\begin{eqnarray}
\Delta^G_{\min}(\ket{\Psi}):= \min_{X}\Delta_{X|X^{\mathsf{c}}}(\ket{\Psi})
\end{eqnarray}
for all $X\subset \{A_1,A_2,\cdots, A_n\}$, where $X$ and $X^c$ denotes a bipartition of $\{A_1,A_2, \cdots, A_n\}$.

\end{definition}

Similar to the ergotropic gaps of work extraction and work injection,  $\Delta^G_{\min}(\cdot)$ provides a genuine measure for $n$-partite entangled states \cite{Guo2022,Puliyil2022}. For tripartite states it equals to double of genuinely multipartite concurrence (GMC) \cite{Ma2011}. The maximally entangled GHZ state \cite{GHZ} yields to the maximum value of $\Delta^G_{\min}$ for any $n$-qubit system.

\begin{definition}
The average biseparable capacity gap (ABCG) of pure state $\ket{\Psi}$ on Hilbert space $\mathcal{H}=\otimes_{i=1}^n\mathcal{H}_{A_i}$ is defined by
\begin{eqnarray}
\Delta^G_{\footnotesize\rm avg}(\ket{\Psi}):=\alpha\Gamma\left(\prod_{X} \Delta_{X|X^\mathsf{c}}(\ket{\Psi})\right)\sum_X \Delta_{X|X^\mathsf{c}}(\ket{\Psi})
\nonumber
\\
\end{eqnarray}
for $X\subset \{A_1, A_2,\cdots, A_n\}$, where $\Gamma(x)=0$ for $x=0$, and $\Gamma(x)=1$ otherwise. $\alpha$ denotes a nonzero constant for faithful entanglement measure.

\end{definition}

Similar to the average ergotropic gap of maximum work extraction \cite{Puliyil2022}, both entanglement measures turn out to be genuine, faithful, LOCC monotone, and able to distinguish tripartite GHZ and W states \cite{Acin2000}. $\Delta^G_{\rm avg}$ and $\Delta_{\min}^G$ are independent measures and may fail to distinguish some genuinely entangled states \cite{Puliyil2022}.

Similar to the biseparable ergotropic gaps of work extraction, the battery capacity gap satisfies the following polygon inequality as
\begin{eqnarray}
\Delta_{X|YZ}\leq\Delta_{Y|ZX}+\Delta_{Z|XY}
\end{eqnarray}
for $X,Y,Z\in\{A,B,C\}$ because the ergotropic gaps of work injection satisfies the same inequality. Inspired by Refs. \cite{Xie2021,Puliyil2022,Jin2023}, we define the battery capacity fill for tripartite states as follows.

\begin{definition}
The battery capacity fill (WCF) is defined by
\begin{equation}
\Delta^G_F(\ket{\Psi}_{ABC}):=\left[\frac{1}{3}Q\!\!\prod_{X\in \{A,B,C\}}\left(Q-\Delta_{X|X^\mathsf{C}}\right)\right]^{1/2},
\end{equation}
where $Q=\sum_{X\in \{A,B,C\}}\Delta_{X|X^\mathsf{C}}$.

\end{definition}

It is obvious that WCF is zero for all product states. It is also faithful, i.e., non zero for genuinely entangled states. It can be used for distinguishing states like GHZ and W states \cite{Acin2000,Puliyil2022}.

\begin{definition}
The battery capacity volume (WCV) of an $n$-qubit pure state is given by
\begin{eqnarray}
\Delta^G_V(\ket{\Psi}):= \left(\prod_{X}\Delta_{X|X^\mathsf{c}}(\ket{\Psi})\right)^{1/N} ,
\end{eqnarray}
for $X\subset \{A_1,A_2,\cdots, A_n\}$, where $N$ denotes the combination number of all bipartition of $\{A_1,A_2,\cdots, A_n\}$.

\end{definition}

WCV is genuine, faithful, and LOCC monotone, as all the biseparable battery capacity gaps are LOCC monotone \cite{Puliyil2022}. Moreover, WCF is inequivalent to MBWCG, ABCG, and WCV by proofs similar to the ones for the ergotropic gap of work extraction \cite{Puliyil2022}.

\textit{Example S4}. Consider the two-qubit pure state
$\ket{\psi}=\sqrt{\lambda}\ket{00}+\sqrt{1-\lambda}\ket{11}$. The battery capacity gap with respect to the local Hamiltonians $\hat{H}_A=\hat{H}_B=|1\rangle\langle1|$ is given by
\begin{eqnarray}
\Delta_{A|B}(\ket{\psi};\hat{H})=4 \left(1-\max \{\lambda, 1-\lambda\} \right).
\end{eqnarray}
This can be extended to mixed states by the standard method \cite{HHH} as
\begin{equation}
    \Delta_{A|B}(\hat\varrho,\ham)=\min \sum_i p_i \Delta_{A|B}(|\psi_i\rangle ; \ham),
\end{equation}
where the minimum over all possible pure state decompositions, $\varrho=\sum_i p_i|\psi_i\rangle\langle\psi_i|$. In a similar manner with the proof of Eq.(45) in ref.~\cite{Yang2022}, we obtain
\begin{eqnarray}
\Delta_{A|B}(\hat\varrho;\hat{H})=2 \left[ 1-\sqrt{1-C^2(\hat\varrho)} \right],
\label{C-capacity}
\end{eqnarray}
where $C(\hat\varrho)$ denotes the concurrence. Consider the isotropic state
\begin{eqnarray}
\hat\varrho_v= \frac{1}{3}[(1-v)\mathbbm{1}_4+(4v-1)\ketbra{\psi}{\psi}]
\end{eqnarray}
with the governing Hamiltonian $\ham=|1\rangle\langle1|$,
where $\mathbb{I}_n$ denotes the $n \times n$ identity matrix and $|\Psi^{-}\rangle=(|01\rangle-|10\rangle) / \sqrt{2}$.
Equation \eqref{C-capacity} and  $C(\varrho_v)=2v-1$ yield a positive gap given by
\begin{eqnarray}
\Delta_{A|B}(\hat\varrho_v;\ham)=2(1-2\sqrt{v-v^2})
\end{eqnarray}
 for $1/2<v\leq 1$ \cite{Hiroshima2000}.

\textit{Example S5}. Consider the generalized tripartite GHZ state  $\ket{\phi}=\cos\theta \ket{000}+\sin\theta\ket{111}$ \cite{GHZ} with $\theta\in (0,\pi/4]$.
For any bipartition, we obtain the gap $\Delta_{A|BC}(\ket{\phi};\ham)=4\sin^2\theta$. The symmetry of the state then implies the genuine multipartite entanglement measure \cite{Guo2022,Puliyil2022}
\begin{eqnarray}
\Delta^G_{\min}(\ket{\phi};\ham)&:= \min_{X}\Delta_{X|X^{\mathsf{c}}}(\ket{\phi};\ham)
\nonumber\\
&=4\sin^2\theta>0,
\end{eqnarray}
minimising over all $X\subset \{A, B, C \}$.

\section{Case study: two two-level atoms in a cavity}

Here we provide details on the example battery model in the main text, which consists of two two-level atoms interacting with a single cavity mode via the resonant Tavis-Cummings Hamiltonian in the rotating wave approximation,
\begin{equation}
\ham=E \left( \hat{a}^\dagger\hat{a} + \frac{\hat\sigma_1^z + \hat\sigma_2^z}{2} \right)+ \hbar g\sum_{i=1}^2 \left( \hat{a}\hat\sigma_i^+ + h.c. \right).
\end{equation}
Here, $\hat a$ denotes the ladder operator of the cavity mode, $\hat \sigma_i^z = |e\rangle_i\langle e| - |g\rangle_i\langle g|$ the Pauli-z matrix of the $i$-th atom, and $\hat\sigma_i^+ = |e\rangle_i\langle g|$ the $i$-th excitation operator. Given that the free part of the Hamiltonian proportional to $E$ and the coupling term commute on resonance, the associated unitary time evolution can be factorized into
\begin{equation}
    e^{-i\ham t/\hbar} = e^{-iE \hat a^\dagger \hat a t/\hbar} \hat U_0^\dagger (t) \hat U (t).
\end{equation}
The unitary $\hat U_0^\dagger (t) = e^{-iE (\hat\sigma_1^z + \hat\sigma_2^z) t/2\hbar}$ describes the free battery evolution, while $\hat U(t)$ represents the time evolution in the interaction picture, which can be given as an operator-valued matrix in the atomic basis of  $\{|gg\rangle, |ge\rangle, |eg\rangle, |ee\rangle\}$ as
 \cite{Kim,Restrepo}
\begin{equation}
\hat U(t)=\left(
\begin{array}{ccccc}
2\hat{a}\hat{\Gamma}\hat{a} & -i\hat{a}\hat{\Xi} & -i\hat{a}\hat{\Xi} & 2\hat{a}\hat{\Gamma}\hat{a}^\dag+\mathbbm{1}
\\
-i\hat{\Xi}\hat{a}  & \frac{1}{2}\hat{\Upsilon}^{-1}\hat{\Gamma} & \frac{1}{2}\hat{\Theta} &  -i\hat{\Xi}\hat{a}^\dag
\\
  -i\hat{\Xi}\hat{a}& \frac{1}{2}\hat{\Theta}
& \frac{1}{2}\hat{\Upsilon}^{-1}\hat{\Gamma} &   -i\hat{\Xi}\hat{a}^\dag
\\
2\hat{a}^\dag\hat{\Gamma}\hat{a}+\mathbbm{1}  & -i\hat{a}^\dag\hat{\Xi} & -i\hat{a}^\dag\hat{\Xi} &  2\hat{a}^\dag\hat{\Gamma}\hat{a}^\dag
\end{array}
\right) .
\end{equation}
Here, $\hat{\Gamma}=\hat{\Upsilon}(\cos(\hat{\Omega}gt)-\mathbbm{1})$,  $\hat{\Theta}=\cos(\hat{\Omega}gt)+\mathbbm{1}$, $\hat{\Xi}=\hat{\Omega}^{-1}\sin(\hat{\Omega}gt)$,  $\hat{\Omega}=\hat{\Upsilon}^{-1/2}=\sqrt{4\hat{a}^\dag\hat{a}+2\mathbbm{1}}$, and $\mathbbm{1}$ denotes the identity operator.

For an initial product state $\hat \rho (0) = \hat \varrho_0 \otimes \hat \rho_c$, the time evolution of the reduced battery state reads as
\begin{eqnarray}\label{eq:rho_TC_evo}
    \hat \varrho(t) &=& \Tr_c \left[  e^{-i\ham t/\hbar} \hat \varrho_0 \otimes \hat \rho_c  e^{i\ham t/\hbar} \right] \\
    &=& \hat U_0 (t) \sum_{n=0}^\infty \langle n|\hat U(t) \hat \varrho_0 \otimes \hat \rho_c \hat U^\dagger (t) |n \rangle \hat U^\dagger_0 (t),  \nonumber
\end{eqnarray}
where we can ignore the free evolution term since it affects neither of the state properties we are interested in.

In the main text, we focus on the case in which the cavity is initially in a thermal state or coherent state. Complementing the results shown there, we here provide additional simulation data for different parameters and initial battery states.

\begin{figure}
\begin{center}
\includegraphics[width=\columnwidth]{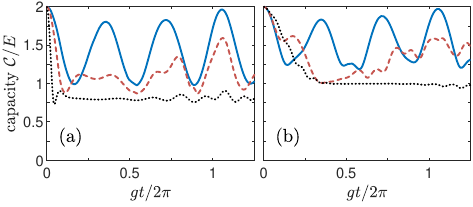}
\caption{Two-atom battery capacity as a function of time for (a) a thermal and (b) a coherent cavity state of average photon numbers $n_0=0.1$ (solid), $n_0=1$ (dashed), and $n_0=10$ (dotted). This plot assumes the initial pure (maximum-capacity) battery state $|eg\rangle$.
}\label{figs1}
\end{center}
\end{figure}

Figure \ref{figs1} plots the battery capacity over time for various average photon numbers $n_0$, comparing (a) an initially thermal to (b) a coherent cavity state. In both cases, higher $n_0$ generally lead to a stronger initial decrease of capacity and suppression of transient oscillations. In this example, the battery is initialized in the state $\hat \varrho(0) = |eg\rangle\langle eg|$.

\begin{figure}
\begin{center}
\includegraphics[width=\columnwidth]{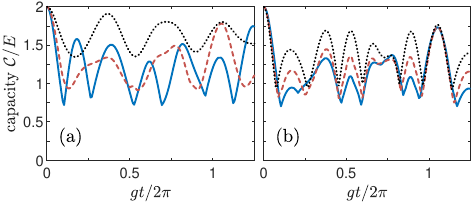}
\caption{Two-atom battery capacity over time for different initial states: (a) pure product states $|ee\rangle$ (solid), $|eg\rangle$ (dashed, same as $|ge\rangle$), and $|gg\rangle$ (dotted). (b) entangled states $|\psi\rangle = \cos\theta|eg\rangle+\sin\theta|ge\rangle$ with $\theta=\pi/16$ (solid), $ \theta=\pi/8$ (dashed), and $\theta=\pi/4$ (dotted). Here, $n_0=0.5$.}\label{figs2}
\end{center}
\end{figure}

Figure \ref{figs2} shows how the time-evolved capacity depends on the initial battery state. In (a), we compare initial pure product states and find that the greatest capacities are attained for $|gg\rangle$, while $|ee\rangle$ exhibits fast oscillations. In (b), we compare entangled initial states of the form $\cos\theta|eg\rangle+\sin\theta|ge\rangle$ for different $\theta$. More entangled states with $\theta$ close to $\pi/4$ exhibit greater transient capacity values.

\begin{figure}
\begin{center}
\includegraphics[width=\columnwidth]{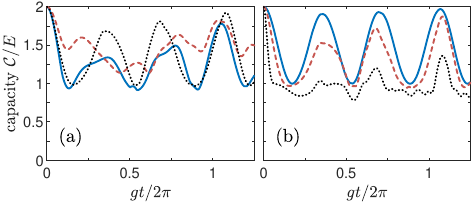}
\caption{(a) Two-atom battery capacity over time for different initial cavity states of the same average photon number $n_0=0.5$: thermal mixture (solid), coherent state (dashed), and squeezed state (dotted). (b) Capacity over time for an initially squeezed cavity state of mean photon numbers $n_0 = \sinh^2 r = 0.1$ (solid), $1$ (dashed), and $10$ (dotted). The battery is initialized in $|eg\rangle$.}
\label{figs3}
\end{center}
\end{figure}

Figure \ref{figs3} (a) compares the time evolution of the battery capacity for different cavity states of the same average photon number $n_0$. Apart from the thermal state, we also consider  $\rho_c = |\psi \rangle \langle \psi|$ with $\langle n|\psi\rangle = e^{-n_0/2} n_0^{n/2}/\sqrt{n!}$  (coherent state) and $ \langle 2n|\psi\rangle =\sqrt{\textrm{sech}(r)}[\sqrt{(2n)!}/n!]2^{-n}\tanh^n(r)$ (squeezed vacuum state) \cite{Delmonte2021}. The three states result in notable differences, with the squeezed state reaching the highest transient capacities. In (b), we compare different degrees of squeezing: higher $r$-values lead to higher effective $n_0 = \sinh^2 r$ and thus lower transient capacities. However, the transient oscillations are more pronounced compared to thermal states of the same $n_0$ in Fig.~\ref{figs1} (b).


\begin{thebibliography}{00}
\bibitem{Perarnau2015} M. Perarnau-Llobet, K. V. Hovhannisyan, M. Huber, P. Skrzypczyk, N. Brunner, and A. Ac\'{\i}n, Extractable work from correlations, Phys. Rev. X 5, 041011 (2015).

\bibitem{Vin2016} S. Vinjanampathy, and J. Anders, Quantum thermodynamics, Contemp. Phys. 57, 545 (2016).

\bibitem{Ciampini2017} M. A. Ciampini,  L. Mancino,  A. Orieux,  C. Vigliar,  P. Mataloni,  M. Paternostro, and  M. Barbieri, Experimental extractable work-based multipartite separability criteria, npj Quant. Inf. 3, 10 (2017).

\bibitem{Francica2017} G. Francica, J. Goold,  F. Plastina, and  M. Paternostro,  Daemonic ergotropy: Enhanced work extraction from quantum correlations, npj Quant. Inf. 3, 12 (2017).

\bibitem{Andolina2019} G. M. Andolina, M. Keck, A. Mari, M. Campisi, V. Giovannetti, and M. Polini, Extractable work, the role of correlations, and asymptotic freedom in quantum batteries, Phys. Rev. Lett. 122, 047702 (2019).

\bibitem{Monsel2020} J. Monsel, M. Fellous-Asiani, B. Huard, and A. Auff\`{e}ves, The energetic cost of work extraction, Phys. Rev. Lett. 124, 130601 (2020).

\bibitem{Opatrny2021} T. Opatrny, A. Misra, and G. Kurizki, Work generation from thermal noise by quantum phase-sensitive observation, Phys. Rev. Lett. 127, 040602 (2021).

\bibitem{Allahverdyan2004}A. E. Allahverdyan, R. Balian, and Th. M. Nieuwenhuizen, Maximal work extraction from finite quantum systems, Europhys. Lett. 67, 565 (2004).

\bibitem{Binder2015} F. C. Binder, S. Vinjanampathy, K. Modi, and J. Goold, Quantacell: powerful charging of quantum batteries, New J. Phys. 17, 075015 (2015).

\bibitem{Campaioli2017} F. Campaioli, F. A. Pollock, F. C. Binder, L. C\'{e}leri, J. Goold, S. Vinjanampathy, and K. Modi, Enhancing the charging power of quantum batteries, Phys. Rev. Lett. 118, 150601 (2017).

\bibitem{Campaioli2018} F. Campaioli,  F. A. Pollock, and  S. Vinjanampathy, \textit{Thermodynamics in the quantum regime}, Springer, Cham, 2018, pp. 207-225.

\bibitem{Rossini2020} D.~Rossini, G.~M.~Andolina, D.~Rosa, M.~Carrega, and M.~Polini, Quantum Advantage in the Charging Process of Sachdev-Ye-Kitaev Batteries, Phys.~Rev.~Lett.~125, 236402 (2020).

\bibitem{Salvia2021} R. Salvia and V. Giovannetti, On the distribution of the mean energy in the unitary orbit of quantum states, Quantum 5, 514 (2021).

\bibitem{Seah2021}S. Seah, M. Perarnau-Llobet, G. Haack, N. Brunner, and S. Nimmrichter, Quantum speed-up in collisional battery charging, Phys. Rev. Lett. 127, 100601(2021).

\bibitem{Shaghaghi2022a}V.~Shaghaghi, V.~Singh, G.~Benenti, and D.~Rosa, \textit{Micromasers as quantum batteries}, Quantum Sci.~Technol.~7, 04LT01 (2022).

\bibitem{Salvia2022} R. Salvia, M. Perarnau-Llobet, G. Haack, N. Brunner, and S. Nimmrichter, Quantum advantage in charging cavity and spin batteries by repeated interactions, arXiv:2205.00026 (2022).

\bibitem{Francica2022} G. Francica, Quantum correlations and ergotropy, Phys. Rev. E 105,  L052101(2022).

\bibitem{Rodriguez2023} C.~Rodriguez, D.~Rosa, and J.~Olle, AI-discovery of a new charging protocol in a micromaser quantum battery, arXiv:2301.09408 (2023).

\bibitem{PS1977}W. Pusz and S. L. Woronowicz. Passive states and KMS states for general quantum systems, Commun. Math. Phys. 58, 273-290 (1977).

\bibitem{PS1978} A. Lenard, Thermodynamical proof of the Gibbs formula for elementary quantum systems, J. Stat. Phys. 19, 575 (1978).

\bibitem{Neumann}J. von Neumann, Thermodynamik quantummechanischer Gesamtheiten, Gott. Nach. 1, 273-291(1927).

\bibitem{Tsallis}C. Tsallis, Possible  generalization of Boltzmann-Gibbs statistics,  J. Stat. Phys. 52, 479 (1988).

\bibitem{HHH} R. Horodecki, P. Horodecki, M. Horodecki, and K. Horodecki, Quantum entanglement, Rev. Mod. Phys. 81, 865 (2009).

\bibitem{Streltsov}A. Streltsov, G. Adesso, and M. B. Plenio, Colloquium: Quantum coherence as a resource, Rev. Mod. Phys. 89, 041003 (2017).

\bibitem{Bu}K. Bu, U. Singh, S.-M. Fei, A. K. Pati, and J. Wu, Maximum relative entropy of coherence: an operational coherence measure, Phys. Rev. Lett. 119, 150405(2017).

\bibitem{Baumgratz} T. Baumgratz, M. Cramer, and M. B. Plenio, Quantifying coherence, Phys. Rev. Lett. 113, 140401  (2014).

\bibitem{Aberg}J. Aberg, Quantifying superposition, arXiv:quant-ph/0612146, 2006.

\bibitem{South}K. Southwell, Quantum coherence, Nature 453, 1003 (2008).


\bibitem{Delmonte2021}A. Delmonte, A. Crescente, M. Carrega, D. Ferraro, and M. Sassetti, Characterization of a Two-Photon Quantum Battery: Initial Conditions, Stability and Work Extraction, Entropy 23, 612 (2021).

 \bibitem{Mukherjee2016}A. Mukherjee, A. Roy, S. S. Bhattacharya, and M. Banik, Presence of quantum correlations results in a nonvanishing ergotropic gap, Phys. Rev. E 93, 052140 (2016).

\bibitem{Nielsen2000} M. A. Nielsen and I. L. Chuang, \textit{Quantum Computation and Quantum Information}. Cambridge University Press, Cambridge, 2000.

\bibitem{Alicki2013}  R. Alicki and M. Fannes, Entanglement boost for extractable work from ensembles of quantum batteries, Phys. Rev. E 87, 042123 (2013).

\bibitem{Hovhannisyan2013} K. V. Hovhannisyan, M. Perarnau-Llobet, M. Huber, and A. Acin, Entanglement generation is not necessary for optimal work extraction, Phys. Rev. Lett. 111, 240401(2013).

\bibitem{Alimuddin2019} M. Alimuddin, T. Guha, and P. Parashar, Bound on ergotropic gap for bipartite separable states, Phys. Rev. A 99, 052320 (2019).

\bibitem{Alimuddin2020}  M. Alimuddin, T. Guha, and P. Parashar, Independence of work and entropy for equal-energetic finite quantum systems: Passive-state energy as an entanglement quantifier, Phys. Rev. E 102, 012145 (2020).

\bibitem{Gyhm2022} J.-Y.~Gyhm, D.~\v{S}afr\'{a}nek, and D.~Rosa, Quantum Charging Advantage Cannot Be Extensive without Global Operations, Phys.~Rev.~Lett.~128, 140501 (2022).

\bibitem{Puliyil2022}  S. Puliyil,  M. Banik, M.  Alimuddin, Thermodynamic signatures of genuinely multipartite entanglement, Phys. Rev. Lett. 129, 070601(2022).

\bibitem{Imai2023} S.~Imai, O.~G\"{u}hne, and S.~Nimmrichter, Work fluctuations and entanglement in quantum batteries, Phys.~Rev.~A 107, 022215 (2023)

\bibitem{Shaghaghi2022}  V. Shaghaghi, G. M.  Palma, G. Benenti, Extracting work from random collisions: A model of a quantum heat engine, Phys. Rev. E 105, 034101 (2022).

\bibitem{Biswas2022} T. Biswas, M. {\L}obejko,  P. Mazurek, and M. Horodecki, Extraction of ergotropy: free energy bound and application to open cycle engines, Quantum 6, 841 (2022).

\bibitem{Tirone2022}S. Tirone , R. Salvia, S. Chessa, and V. Giovannetti, Quantum work capacitances, arXiv:2211.02685, 2022.

\bibitem{Alimuddin2020b} M. Alimuddin, T. Guha, and P. Parashar, Structure of passive states and its implication in charging quantum batteries, Phys. Rev. E 102, 022106 (2020).

\bibitem{Marshall2010} A. W. Marshall, I. Olkin, and B. C. Arnold, \textit{Inequalities: Theory of Majorization and Its Applications}, 2nd ed, Springer, New York, 2010.

\bibitem{Napoli}C. Napoli, T. R. Bromley, M. Cianciaruso, M. Piani, N. Johnston, and G. Adesso, Robustness of coherence: an operational and observable measure of quantum coherence, Phys. Rev. Lett. 116, 150502 (2016).

\bibitem{Yang2022} X. Yang,  Y. H. Yang, L. M. Zhao,  and M. X. Luo, A new entanglement measure based dual entropy, arXiv:2204.07407, 2022.

\bibitem{Hiroshima2000} T. Hiroshima and S. Ishizaka, Local and nonlocal properties of Werner states, Phys. Rev. A 62, 044302 (2000).

\bibitem{GHZ}D.M. Greenberger, M. A. Horne, and A. Zeilinger, in \textit{Bell's Theorem, Quantum Theory, and Conceptions of the Universe}, edited by M. Kafatos (Kluwer, Dordrecht, 1989), pp. 69-72.

\bibitem{Guo2022}Y. Guo, Y. Jia, X. Li, and L. Huang, Genuine multipartite entanglement measure, J. Phys. A Math. Theor. 55, 145303 (2022).

\bibitem{Werner}R. F. Werner, Quantum states with Einstein-Podolsky-Rosen correlations admitting a hidden-variable model, Phys. Rev. A 40, 4277 (1989).


\bibitem{Acin2000}A. Ac\'{\i}n, A. Andrianov, L. Costa, E. Jan, J. I. Latorre, and R. Tarrach, Generalized schmidt decomposition and classification of three-quantum-bit states, Phys. Rev. Lett. 85, 1560 (2000).

\bibitem{Ma2011}Z. H. Ma, Z. H. Chen, J. L. Chen, C. Spengler, A. Gabriel, and M. Huber, Measure of genuine multipartite entanglement with computable lower bounds, Phys. Rev. A 83, 062325 (2011).

\bibitem{Xie2021}S. Xie and J. H. Eberly, Triangle measure of tripartite entanglement, Phys. Rev. Lett. 127, 040403 (2021).

\bibitem{Jin2023}Z. X. Jin, Y. H. Tao, Y. T. Gui, S. M. Fei, X. Li-Jost, and X. F. Qiao, Concurrence triangle induced genuine multipartite entanglement measure, Results in Phys. 44, 106155 (2023).

\bibitem{Kim}M. S. Kim, J. Lee, D. Ahn, and P. L. Knight, Entanglement induced by a single-mode heat environment, Phys. Rev. A 65, 040101(R) (2002).

\bibitem{Restrepo}J. Restrepo and B. A. Rodriguez, Dynamics of entanglement and quantum discord in the Tavis-Cummings model, J. Phys. B: At. Mol. Opt. Phys. 49, 125502 (2016).


\bibitem{Safranek2022}D. \v{S}afranek, D. Rosa, and Felix C. Binder, Work extraction from unknown quantum sources, Phys. Rev. Lett. 130, 210401 (2023).

\bibitem{Salamon2020}S. Juli\`{a}-Farr\'{e}, T. Salamon, A. Riera, M. N. Bera, and M. Lewenstein, Bounds on the capacity and power of quantum batteries, Phys. Rev. Research 2, 023113 (2020).

\bibitem{Zheng}W. Zheng, Z. Ma, H. Wang, S. M. Fei, and X.  Peng, Experimental demonstration of observability and operability of robustness of coherence, Phys. Rev. Lett. 120, 230504 (2018).

\bibitem{Franc2020}G. Francica, F. C. Binder, G. Guarnieri, M. T. Mitchison, J. Goold, and F. Plastina, Quantum coherence and ergotropy, Phys. Rev. Lett. 125, 180603 (2020).

\bibitem{Wootters}W. K. Wootters, Entanglement of formation of an arbitrary state of two qubits, Phys. Rev. Lett. 80, 2245 (1998).

\end{thebibliography}
\end{document}